\documentclass[superscriptaddress]{revtex4-2}
\usepackage{amsmath,amssymb,amsthm,graphicx,xcolor}
\usepackage{braket,bm,bbm}
\newtheorem{theorem}{Theorem}
\newtheorem{definition}{Definition}
\newtheorem{proposition}{Proposition}
\newtheorem{example}{Example}
\DeclareMathOperator{\Tr}{Tr}
\DeclareMathOperator{\PI}{PI}
\DeclareMathOperator{\PT}{PT}
\begin{document}
\title{Postmeasurement information and nonlocality of quantum state discrimination}
\author{Jinhyeok Heo}
\affiliation{Department of Mathematics, Kyung Hee University, Seoul 02447, Republic of Korea}
\author{Donghoon Ha}
\author{Jeong San Kim}
\email{freddie1@khu.ac.kr}
\affiliation{Department of Applied Mathematics and Institute of Natural Sciences, Kyung Hee University, Yongin 17104, Republic of Korea}
\begin{abstract}
In quantum state discrimination, nonlocality arises when optimal discrimination cannot be achieved using only local operations and classical communication. It has recently been shown that postmeasurement information identifying the subensemble containing the prepared state can either annihilate or create such nonlocality. Here, we demonstrate that these effects can depend critically on how the original ensemble is partitioned into subensembles. We derive sufficient conditions under which different ensemble partitions lead to the annihilation or creation of nonlocality upon the revelation of postmeasurement information. We further construct explicit bipartite quantum state ensembles satisfying these conditions, thereby illustrating the partition-dependent nature of nonlocality in quantum state discrimination.
\end{abstract}
\maketitle

\section{Introduction}\label{sec:int}
Quantum nonlocality is a fundamental feature of multipartite quantum systems. A paradigmatic manifestation of such nonlocality is quantum entanglement, which cannot be generated from separable states using only \emph{local operations and classical communication} (LOCC) \cite{horo2009,chit2014}. Entanglement serves as an essential resource for various tasks in multipartite quantum information processing, including quantum teleportation and quantum communication \cite{benn1992,benn1993}. These developments have also motivated extensive efforts to characterize other forms of quantum nonlocality, such as Bell nonlocality and quantum steering \cite{brun2014,uola2020}.

Nonlocality also arises in the discrimination of multipartite quantum states. For a given ensemble of quantum states, nonlocality is said to occur when globally optimal discrimination cannot be achieved using only LOCC \cite{chil2013}. In general, mutually orthogonal quantum states can be perfectly discriminated by global measurements, whereas nonorthogonal states cannot \cite{chef2000,barn2009,berg2010,bae2015}. Nevertheless, some sets of multipartite orthogonal quantum states cannot be perfectly discriminated using only LOCC \cite{benn19991,ghos2001,walg2002}. Moreover, certain ensembles of nonorthogonal quantum states cannot be optimally discriminated by LOCC \cite{pere1991,duan2007,chit2013}.
Furthermore, LOCC can activate nonlocality in state discrimination by transforming certain locally distinguishable ensembles into locally indistinguishable ones \cite{band2021,ghos2022,li2022,gupt2023}.

Recently, it was shown that nonlocality in quantum state discrimination can be either annihilated or created when \emph{postmeasurement information} (PI) identifying the subensemble containing the prepared state is provided \cite{ball2008,gopa2010,ha20221,ha20222}. Some ensembles that exhibit nonlocality without PI become optimally distinguishable by LOCC once the prepared subensemble is revealed. In this case, PI annihilates the nonlocality in state discrimination. Conversely, there exist ensembles that can be optimally discriminated by LOCC without PI but cease to be optimally distinguishable by LOCC when the prepared subensemble is revealed \cite{ha20221,ha20222}. In this case, PI creates nonlocality. These contrasting effects naturally raise the question of whether PI about the prepared subensemble can always be used to annihilate or create nonlocality in quantum state discrimination.

Here, we consider bipartite quantum state discrimination and show that the annihilation or creation of nonlocality by PI can depend on how the original ensemble is partitioned into subensembles. In particular, even when PI associated with one ensemble partition annihilates or creates nonlocality, PI associated with another partition need not have the same effect. We establish sufficient conditions for this partition-dependent annihilation or creation of nonlocality (Theorems~\ref{thm:alc} and \ref{thm:aulc}) and construct bipartite quantum state ensembles satisfying these conditions (Examples~\ref{ex:anpi} and \ref{ex:crpi}).

This paper is organized as follows. In Sec.~\ref{sec:pre}, we review the definitions and relevant properties of separable measurements and entanglement witnesses used to analyze nonlocality in quantum state discrimination. We also introduce \emph{minimum-error discrimination} (ME) and present several useful properties of nonlocality in ME. In Sec.~\ref{sec:qsdpi}, we review ME with postmeasurement information and establish sufficient conditions for nonlocality in this setting (Theorems~\ref{thm:scmp} and \ref{thm:iepi}). In Sec.~\ref{sec:lun}, we introduce the notions of annihilating and creating nonlocality by PI and derive sufficient conditions under which these effects depend on the choice of subensembles specified by PI (Theorems~\ref{thm:alc} and \ref{thm:aulc}). We then construct four-state ensembles for which nonlocality is annihilated or created depending on how the ensemble is partitioned into two-state subensembles (Examples~\ref{ex:anpi} and \ref{ex:crpi}). Finally, in Sec.~\ref{sec:dis}, we summarize our results and discuss directions for future research.

\section{Preliminary}\label{sec:pre}
For a bipartite Hilbert space $\mathcal{H}=\mathbb{C}^{d_{\rm A}}\otimes\mathbb{C}^{d_{\rm B}}$, let $\mathbb{H}$ be the set of all Hermitian operators acting on $\mathcal{H}$.
We denote by $\mathbb{H}_{+}$ the set of all positive-semidefinite operators in $\mathbb{H}$, that is,
\begin{equation}\label{eq:spso}
\mathbb{H}_{+}=\{E\in\mathbb{H}\,|\,\bra{v}E\ket{v}\geqslant0~\forall\ket{v}\in\mathcal{H}\}.
\end{equation}
A bipartite quantum state is described by a density operator $\rho$, that is, a positive-semidefinite operator $\rho\in\mathbb{H}_{+}$ with unit trace $\Tr\rho=1$.
A measurement is represented by a positive operator-valued measure $\{M_{i}\}_{i}$, that is, a set of positive-semidefinite operators $M_{i}\in\mathbb{H}_{+}$ satisfying 
the completeness relation $\sum_{i}M_{i}=\mathbbm{1}$, where $\mathbbm{1}$ is the identity operator in $\mathbb{H}$. 
For the  state $\rho$, the probability of obtaining the measurement outcome with respect to $M_{j}$ is $\Tr(\rho M_{j})$.

\begin{definition}\label{def:sep}
$E\in\mathbb{H}_{+}$ is called \emph{separable} if it can be described as
\begin{equation}\label{eq:sepod}
E=\sum_{l}A_{l}\otimes B_{l},
\end{equation}
where $A_{l}$ and $B_{l}$ are positive-semidefinite operators acting on $\mathbb{C}^{d_{\rm A}}$ and $\mathbb{C}^{d_{\rm B}}$ of $\mathcal{H}$, respectively.
\end{definition}

\subsection{Entanglement witness}\label{ssec:ew}
\noindent We denote the set of all \emph{separable} operators in $\mathbb{H}_{+}$ as
\begin{equation}\label{eq:sepdef}
\mathbb{SEP}=\{E\in\mathbb{H}_{+}\,|\, E:\mbox{separable}\},
\end{equation}
and its dual set as $\mathbb{SEP}^{*}$, that is,
\begin{equation}\label{eq:sepsd}
\mathbb{SEP}^{*}=\{E\in\mathbb{H}\,|\,\Tr(EF)\geqslant0~\forall F\in\mathbb{SEP}\}.
\end{equation}
An element in $\mathbb{SEP}^{*}$ is also called \emph{block positive}.

A measurement $\{M_{i}\}_{i}$ is called a \emph{separable measurement} if $M_{i}\in\mathbb{SEP}$ for all $i$, and a measurement is called a \emph{LOCC measurement} if it can be realized by LOCC. Note that every LOCC measurement is a separable measurement\cite{chit2014}.

\begin{definition}\label{def:ew}
$W\in\mathbb{H}$ is called an \emph{entanglement witness} (EW)
if $\Tr(\sigma W)\geqslant0$ for any state $\sigma$ in $\mathbb{SEP}$ but $\Tr(\rho W)<0$ for some state $\rho$ in $\mathbb{H}_{+}\backslash\mathbb{SEP}$, or equivalently
\begin{equation}\label{eq:wsshp}
W\in\mathbb{SEP}^{*}\backslash\mathbb{H}_{+}.
\end{equation}
\end{definition}

Throughout this paper, we only consider the situation of discriminating \emph{bipartite} quantum states from the ensemble of the form,
\begin{equation}\label{eq:ens}
\mathcal{E}=\{\eta_{i},\rho_{i}\}_{i\in\Lambda},~\Lambda=\{1,2,3,4\},
\end{equation}
where the state $\rho_{i}$ is prepared with the probability $\eta_{i}$. 

\subsection{Nonlocality in bipartite quantum state discrimination}\label{ssec:nbqsd}
Let us consider the quantum state discrimination of $\mathcal{E}$ in Eq.~\eqref{eq:ens} using a measurement $\mathcal{M}=\{M_{i}\}_{i\in\Lambda}$. Here, the detection of $M_{i}$ means that 
the prepared state is guessed to be $\rho_{i}$.
The \emph{ME} of $\mathcal{E}$ is to achieve the optimal success probability,
\begin{equation}\label{eq:pgdef}
p_{\rm G}(\mathcal{E})=\max_{\mathcal{M}}\sum_{i\in\Lambda}\eta_{i}\Tr(\rho_{i}M_{i}),
\end{equation}
where the maximum is taken over all possible measurements\cite{hels1969}.

When the available measurements are limited to LOCC measurements, 
we denote the maximum success probability by
\begin{equation}\label{eq:pldef}
p_{\rm L}(\mathcal{E})=\max_{\mathrm{LOCC}\,\mathcal{M}}\sum_{i\in\Lambda}\eta_{i}\Tr(\rho_{i}M_{i}).
\end{equation}
Similarly, we denote the maximum success probability over all possible separable measurements as
\begin{equation}\label{eq:pptdef}
p_{\rm SEP}(\mathcal{E})=\max_{\mathrm{Separable}\,\mathcal{M}}\sum_{i\in\Lambda}\eta_{i}\Tr(\rho_{i}M_{i}).
\end{equation}
From the definitions, we have
\begin{equation}\label{eq:plptpg}
\max\{\eta_{1},\ldots,\eta_{n}\}\leqslant p_{\rm L}(\mathcal{E})\leqslant p_{\rm SEP}(\mathcal{E})\leqslant p_{\rm G}(\mathcal{E}).
\end{equation}

In discriminating the states from the ensemble $\mathcal{E}$, quantum nonlocality occurs if the optimal success probability in Eq.~\eqref{eq:pgdef} cannot be achieved only by LOCC measurements, that is,
\begin{equation}\label{eq:qnoc}
p_{\rm L}(\mathcal{E})<p_{\rm G}(\mathcal{E}).
\end{equation}
From Inequality~\eqref{eq:plptpg}, we can easily see that Inequality~\eqref{eq:qnoc} holds if 
\begin{equation}\label{eq:pgsc}
p_{\rm SEP}(\mathcal{E})<p_{\rm G}(\mathcal{E}).
\end{equation}
Moreover, the following proposition provides a condition for Inequality~\eqref{eq:pgsc} in terms of EW.
\begin{proposition}[\cite{ha2023}]\label{prop:sgi}
For a bipartite quantum state ensemble $\mathcal{E}=\{\eta_{i},\rho_{i}\}_{i\in\Lambda}$ satisfying 
\begin{equation}\label{eq:prco}
\eta_{1}\rho_{1}-\eta_{i}\rho_{i}\in\mathbb{SEP}^{*}
\end{equation}
for all $i\in\Lambda$, we have
\begin{equation}\label{eq:sgi}
p_{\rm SEP}(\mathcal{E})<p_{\rm G}(\mathcal{E})
\end{equation}
if and only if there exists an EW in $\{\eta_{1}\rho_{1}-\eta_{i}\rho_{i}\}_{i\in\Lambda}$.
\end{proposition}

\section{Quantum state discrimination with postmeasurement information}\label{sec:qsdpi}

For a state ensemble $\mathcal{E}$ in Eq.~\eqref{eq:ens} and a two-element subset $S$ of the index set $\Lambda$, let us consider the subensembles,
\begin{equation}\label{eq:suben}
\mathcal{E}_{0}=\Big\{\frac{\eta_{i}}{\sum_{j\in S}\eta_{j}},\rho_{i}\Big\}_{i\in S},~
\mathcal{E}_{1}=\Big\{\frac{\eta_{i}}{\sum_{j\in S^{\mathsf{c}}}\eta_{j}},\rho_{i}\Big\}_{i\in S^{\mathsf{c}}},~S^{\mathsf{c}}=\Lambda\backslash S
\end{equation}
where $S^{\mathsf{c}}$ is the complement of $S$ in $\Lambda$.
For the case that the state $\rho_{i}$ belongs to $\mathcal{E}_{0}$ in Eq.~\eqref{eq:suben}, we note that the preparation of $\rho_{i}$ with probability $\eta_{i}$ from the ensemble $\mathcal{E}$ is equivalent to the preparation of the subensemble $\mathcal{E}_{0}$ with probability $\sum_{j\in S}\eta_{j}$ followed by the preparation of $\rho_{i}$ from $\mathcal{E}_{0}$ with probability $\eta_{i}/\sum_{j\in S}\eta_{j}$.
We denote by $\PI_{S}$ the classical information $b\in\{0,1\}$ about the prepared subensemble $\mathcal{E}_{b}$ defined in Eq.~\eqref{eq:suben}, that is
\begin{equation}\label{eq:pis}
\PI_{S}=\left\{
\begin{array}{rl}
0,&i\in S,\\ 
1,&i\in S^{\mathsf{c}},
\end{array}
\right.
\end{equation}
where $i$ is the index of the prepared state $\rho_{i}$.

Now, let us consider the situation of discriminating the quantum states from $\mathcal{E}$ when $\PI_{S}$ in Eq.~\eqref{eq:pis} is given.
In this situation, a measurement can be represented by a positive operator-valued measure $\tilde{\mathcal{M}}=\{\tilde{M}_{\vec{\omega}}\}_{\vec{\omega}\in\Omega_{S}}$ 
where the outcome space is the Cartesian product,
\begin{equation}
\Omega_{S}=S\times S^{\mathsf{c}}.
\end{equation}
Here, the detection of $\tilde{M}_{(\omega_{0},\omega_{1})}$ means that 
we guess the prepared state as $\rho_{\omega_{0}}$ or $\rho_{\omega_{1}}$ 
according to $\PI_{S}=0$ or $1$, respectively\cite{ball2008,gopa2010}.

\emph{ME of $\mathcal{E}$ with $\PI_{S}$} is to maximize the average probability of correct guessing where the optimal success probability is defined as
\begin{equation}\label{eq:pgpie}
p_{\rm G}^{\PI}(\mathcal{E},S)=
\max_{\tilde{\mathcal{M}}}\Big(
\sum_{i\in S}\eta_{i}\mathrm{Tr}\big[\rho_{i} \sum_{j\in S^{\mathsf{c}}}\tilde{M}_{(i,j)}\big]+\sum_{i\in S^{\mathsf{c}}}
\eta_{i}\mathrm{Tr}\big[\rho_{i} \sum_{j\in S}\tilde{M}_{(j,i)}\big]\Big),
\end{equation}
where the maximum is taken over all possible measurements. 
Note that when $\rho_{i}$ is prepared and $\PI_{S}=b$ is given, the prepared state is correctly guessed if we obtain a measurement outcome $\vec{\omega}\in\Omega_{S}$ with $\omega_{b}=i$.
Figure~\ref{fig:mepi} illustrates ME of $\mathcal{E}$ with $\PI_{S}$.
\begin{figure}[t]
\includegraphics{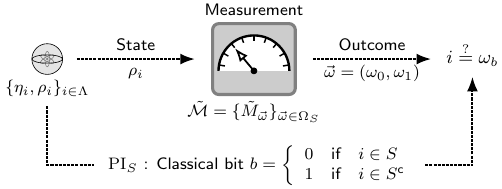}
\caption{
ME of $\mathcal{E}=\{\eta_{i},\rho_{i}\}_{i\in\Lambda}$ with $\PI_{S}$. For each $i\in\Lambda$, the state $\rho_{i}$ is prepared with the probability $\eta_{i}$. After performing a measurement $\tilde{\mathcal{M}}=\{\tilde{M}_{\vec\omega}\}_{\vec\omega\in\Omega_{S}}$, the classical bit $b$ which is $0$ if $i\in S$ and $1$ otherwise is given. 
For each measurement outcome $(\omega_{0},\omega_{1})=\vec\omega\in\Omega_{S}$, the prepared state is guessed to be $\rho_{\omega_{0}}$ or $\rho_{\omega_{1}}$ according to $\PI_{S}=0$ or $1$, respectively. When $\rho_{i}$ is prepared, it is correctly guessed if a measurement outcome $\vec\omega\in\Omega_{S}$ with $\omega_{b}=i$ is obtained; otherwise, errors occur in guessing the prepared state.
}\label{fig:mepi}
\end{figure}

When the available measurements are limited to LOCC measurements, we denote the maximum success probability by
\begin{equation}\label{eq:plpie}
p_{\rm L}^{\PI}(\mathcal{E},S)=
\max_{\mathrm{LOCC}\,\tilde{\mathcal{M}}}\Big(
\sum_{i\in S}\eta_{i}\mathrm{Tr}\big[\rho_{i} \sum_{j\in S^{\mathsf{c}}}\tilde{M}_{(i,j)}\big]
+\sum_{i\in S^{\mathsf{c}}}
\eta_{i}\mathrm{Tr}\big[\rho_{i}\sum_{j\in S}\tilde{M}_{(j,i)}\big]\Big).
\end{equation}
Similarly, we denote 
\begin{equation}\label{eq:pspe}
p_{\rm SEP}^{\PI}(\mathcal{E},S)
=\max_{\mathrm{Separable}\,\tilde{\mathcal{M}}}\Big(
\sum_{i\in S}\eta_{i}\mathrm{Tr}\big[\rho_{i} \sum_{j\in S^{\mathsf{c}}}\tilde{M}_{(i,j)}\big]+\sum_{i\in S^{\mathsf{c}}}
\eta_{i}\mathrm{Tr}\big[\rho_{i} \sum_{j\in S}\tilde{M}_{(j,i)}\big]\Big),
\end{equation}
where the maximum is taken over all possible separable measurements.
From the definitions, we have
\begin{equation}\label{eq:pppi}
p_{\rm L}^{\PI}(\mathcal{E},S)\leqslant 
p_{\rm SEP}^{\PI}(\mathcal{E},S)\leqslant 
p_{\rm G}^{\PI}(\mathcal{E},S).
\end{equation}

We note that for a given measurement $\{\tilde{M}_{\vec{\omega}}\}_{\vec{\omega}\in\Omega}$, the success probability, that is, the right-hand side of Eq.~\eqref{eq:pgpie} without maximization, can be rewritten as
\begin{equation}\label{eq:aplpie}
\sum_{i\in S}\eta_{i}\mathrm{Tr}\big[\rho_{i} \sum_{j\in S^{\mathsf{c}}}\tilde{M}_{(i,j)}\big]+\sum_{i\in S^{\mathsf{c}}}\eta_{i}\mathrm{Tr}\big[\rho_{i} \sum_{j\in S}\tilde{M}_{(j,i)}\big]
=2\sum_{\vec{\omega}\in\Omega_{S}}\tilde{\eta}_{\vec{\omega}}
\mathrm{Tr}(\tilde{\rho}_{\vec{\omega}}\tilde{M}_{\vec{\omega}})
\end{equation}
where 
\begin{equation}\label{eq:avpr}
\tilde{\eta}_{\vec{\omega}}
=\frac{1}{2}\sum_{b\in\{0,1\}}\eta_{\omega_{b}}
\end{equation}
is the average of the probabilities $\eta_{\omega_{0}}$ and $\eta_{\omega_{1}}$, and
\begin{equation}\label{eq:avrh}
\tilde{\rho}_{\vec{\omega}}
=\frac{\sum_{b\in\{0,1\}}\eta_{\omega_{b}}\rho_{\omega_{b}}}{\sum_{b\in\{0,1\}}\eta_{\omega_{b}}}
\end{equation}
is the average of the states $\rho_{\omega_{0}}$ and $\rho_{\omega_{1}}$ for $\vec\omega=(\omega_{0},\omega_{1})$ \cite{ha20221}.
Since $\{\tilde{\eta}_{\vec{\omega}}\}_{\vec{\omega}\in\Omega_{S}}$ and $\{\tilde{\rho}_{\vec{\omega}}\}_{\vec{\omega}\in\Omega_{S}}$ are a probability distribution and a set of states, respectively, Eqs.~\eqref{eq:pgpie},\,\eqref{eq:plpie},\,\eqref{eq:pspe} and \eqref{eq:aplpie} imply
\begin{eqnarray}
p_{\rm G}^{\PI}(\mathcal{E},S)&=&2p_{\rm G}(\tilde{\mathcal{E}}),\nonumber\\
p_{\rm L}^{\PI}(\mathcal{E},S)&=&2p_{\rm L}(\tilde{\mathcal{E}}),\nonumber\\
p_{\rm SEP}^{\PI}(\mathcal{E},S)&=&2p_{\rm SEP}(\tilde{\mathcal{E}}),\label{eq:plpie2}
\end{eqnarray}
where $\tilde{\mathcal{E}}$ is the ensemble consisting of the average states $\tilde{\rho}_{\vec{\omega}}$ prepared with the probabilities $\tilde{\eta}_{\vec{\omega}}$ in Eqs.~\eqref{eq:avpr} and \eqref{eq:avrh},
\begin{equation}\label{eq:tedef}
\tilde{\mathcal{E}}=\{\tilde{\eta}_{\vec{\omega}},\tilde{\rho}_{\vec{\omega}}\}_{\vec{\omega}\in\Omega_{S}}.
\end{equation}

In ME of $\mathcal{E}$ with $\PI_{S}$, quantum nonlocality occurs if the optimal success probability in Eq.~\eqref{eq:pgpie} cannot be achieved only by LOCC measurements, that is,
\begin{equation}\label{eq:defnlwemepi}
p_{\rm L}^{\PI}(\mathcal{E},S)<p_{\rm G}^{\PI}(\mathcal{E},S).
\end{equation}
From Inequality~\eqref{eq:pppi}, we can easily verify that Inequality~\eqref{eq:defnlwemepi} holds if 
\begin{equation}\label{eq:scnpi}
p_{\rm SEP}^{\PI}(\mathcal{E},S)<p_{\rm G}^{\PI}(\mathcal{E},S).
\end{equation}

For a fixed $\vec\mu=(\mu_{0},\mu_{1})\in\Omega_{S}$, the following theorem provides a necessary and sufficient condition for the maximum success probability $p_{\rm SEP}^{\PI}(\mathcal{E},S)$ in Eq.~\eqref{eq:pspe} to be obtained by guessing the prepared state as $\rho_{\mu_{0}}$ or $\rho_{\mu_{1}}$ according to $\rm{PI}_{S}=0$ or $1$, respectively.

\begin{theorem}\label{thm:scmp}
For a bipartite quantum state ensemble $\mathcal{E}=\{\eta_{i},\rho_{i}\}_{i\in\Lambda}$, a two-element subset $S$ of $\Lambda$ and $\vec{\mu}\in\Omega_{S}$,
\begin{equation}\label{eq:scmp}
p_{\rm SEP}^{\PI}(\mathcal{E},S)=2\tilde{\eta}_{\vec{\mu}}
\end{equation}
if and only if 
\begin{equation}\label{eq:scmpc}
\tilde{\eta}_{\vec{\mu}}\tilde{\rho}_{\vec{\mu}}-\tilde{\eta}_{\vec{\omega}}\tilde{\rho}_{\vec{\omega}}\in\mathbb{SEP}^{*}
\end{equation}
for all $\vec{\omega}\in\Omega_{S}$.
\end{theorem}

\begin{proof}
Let us suppose Eq.~\eqref{eq:scmp}.
To show Inclusion~\eqref{eq:scmpc}, we first assume
\begin{equation}\label{eq:nsps}
\tilde\eta_{\vec\mu}\tilde\rho_{\vec\mu}-\tilde\eta_{\vec\omega^{\circ}}\tilde\rho_{\vec\omega^{\circ}}\notin\mathbb{SEP}^{*}
\end{equation}
for some $\vec\omega^{\circ}\in\Omega_{S}$ and lead to a contradiction.
This assumption implies that there exists a pure product state $\ket{v}\!\bra{v}\in\mathbb{SEP}$ satisfying
\begin{equation}\label{eq:psnz}
\bra{v}(\tilde\eta_{\vec\mu}\tilde\rho_{\vec\mu}-\tilde\eta_{\vec\omega^{\circ}}\tilde\rho_{\vec\omega^{\circ}})\ket{v}<0.
\end{equation} 
For a separable measurement $\{\tilde{M}_{\vec\omega}\}_{\vec\omega\in\Omega_{S}}$ with
\begin{equation}\label{eq:smthmo}
{\tilde M}_{\vec\omega}=\left\{
\begin{array}{ccl}
\mathbbm{1}-\ket{v}\!\bra{v}&,&\vec\omega=\vec\mu\\
\ket{v}\!\bra{v}&,&\vec\omega=\vec\omega^{\circ}\\
\mathbb{O}&,&\text{otherwise}
\end{array}
\right.
\end{equation} 
where $\mathbb{O}$ is the zero operator in $\mathbb{H}$, we have 
\begin{eqnarray}
p_{\rm SEP}^{\rm PI}(\mathcal{E},S)
&=&2p_{\rm SEP}(\tilde{\mathcal{E}})\nonumber\\
&\geqslant&2\sum_{\vec\omega\in\Omega_{S}}\tilde\eta_{\vec\omega}\Tr(\tilde\rho_{\vec\omega}\tilde{M}_{\vec\omega})\nonumber\\
&=&2\sum_{\vec\omega\in\Omega_{S}}\Tr\Big[(\tilde\eta_{\vec\mu}\tilde\rho_{\vec\mu}-\tilde\eta_{\vec\mu}\tilde\rho_{\vec\mu}+\tilde\eta_{\vec\omega}\tilde\rho_{\vec\omega})\tilde{M}_{\vec\omega}\Big]\nonumber\\
&=&2\tilde\eta_{\vec\mu}\Tr(\tilde\rho_{\vec\mu}\sum_{\vec\omega\in\Omega_{S}}\tilde{M}_{\vec\omega})
-2\sum_{\vec\omega\in\Omega_{S}}\Tr[(\tilde\eta_{\vec\mu}\tilde\rho_{\vec\mu}-\tilde\eta_{\vec\omega}\tilde\rho_{\vec\omega})\tilde{M}_{\vec\omega}]\nonumber\\
&=&2\tilde\eta_{\vec\mu}-2\bra{v}(\tilde\eta_{\vec\mu}\tilde\rho_{\vec\mu}-\tilde\eta_{\vec\omega^{\circ}}\tilde\rho_{\vec\omega^{\circ}})\ket{v}>2\tilde\eta_{\vec\mu},\label{eq:psctd}
\end{eqnarray}
where the first equality is by Eq.~\eqref{eq:plpie2}, the first inequality is from the definition of $p_{\rm SEP}(\tilde{\mathcal{E}})$ in Eq.~\eqref{eq:pptdef}, the last equality is due to $\sum_{\vec\omega\in\Omega_{S}}\tilde{M}_{\vec\omega}=\mathbbm{1}$ together with Eq.~\eqref{eq:smthmo}, and the last inequality follows from Inequality~\eqref{eq:psnz}.
Inequality \eqref{eq:psctd} contradicts Eq.~\eqref{eq:scmp}, therefore $\tilde{\eta}_{\vec{\mu}}\tilde{\rho}_{\vec{\mu}}-\tilde{\eta}_{\vec{\omega}}\tilde{\rho}_{\vec{\omega}}\in\mathbb{SEP}^{*}$ for all $\vec\omega\in\Omega_{S}$.

Conversely, let us suppose Inclusion~\eqref{eq:scmpc}. For a separable measurement $\tilde{\mathcal{M}}=\{\tilde{M}_{\vec\omega}\}_{\vec\omega\in\Omega_{S}}$ giving $p_{\rm SEP}(\tilde{\mathcal{E}})$, we have 
\begin{eqnarray}
\tilde\eta_{\vec\mu}&\leqslant& p_{\rm SEP}(\tilde{\mathcal{E}})\nonumber\\
&=&\sum_{\vec\omega\in\Omega_{S}}\tilde\eta_{\vec\omega}\Tr(\tilde\rho_{\vec\omega}\tilde{M}_{\vec\omega})\nonumber\\
&=&\sum_{\vec\omega\in\Omega_{S}}\Tr\Big[(\tilde\eta_{\vec\mu}\tilde\rho_{\vec\mu}-\tilde\eta_{\vec\mu}\tilde\rho_{\vec\mu}+\tilde\eta_{\vec\omega}\tilde\rho_{\vec\omega})\tilde{M}_{\vec\omega}\Big]\nonumber\\
&=&\tilde\eta_{\vec\mu}\Tr(\tilde\rho_{\vec\mu}\sum_{\vec\omega\in\Omega_{S}}\tilde{M}_{\vec\omega})-\sum_{\vec\omega\in\Omega_{S}}\Tr[(\tilde\eta_{\vec\mu}\tilde\rho_{\vec\mu}
-\tilde\eta_{\vec\omega}\tilde\rho_{\vec\omega})\tilde{M}_{\vec\omega}]\nonumber\\
&\leqslant&\tilde\eta_{\vec\mu}\Tr(\tilde\rho_{\vec\mu}\sum_{\vec\omega\in\Omega_{S}}\tilde{M}_{\vec\omega})=\tilde\eta_{\vec\mu},\label{eq:slem}
\end{eqnarray}
where the first inequality is by Inequality~\eqref{eq:plptpg}, the first equality is from the assumption of $\tilde{\mathcal{M}}$, the last inequality is due to $\tilde{\eta}_{\vec{\mu}}\tilde{\rho}_{\vec{\mu}}-\tilde{\eta}_{\vec{\omega}}\tilde{\rho}_{\vec{\omega}}\in\mathbb{SEP}^{*}$ and $\tilde{M}_{\vec\omega}\in\mathbb{SEP}$ for all $\vec\omega\in\Omega_{S}$, and the last equality is from $\sum_{\vec\omega\in\Omega_{S}}\tilde{M}_{\vec\omega}=\mathbbm{1}$.
Since Inequality~\eqref{eq:slem} means $p_{\rm SEP}(\tilde{\mathcal{E}})=\tilde\eta_{\vec\mu}$, it follows from Eq.~\eqref{eq:plpie2} that Eq.~\eqref{eq:scmp} is satisfied.
\end{proof}

For an ensemble $\mathcal{E}=\{\eta_{i},\rho_{i}\}_{i\in\Lambda}$ and a two-element subset $S$ of $\Lambda$ satisfying Eq.~\eqref{eq:scmp} for some $\vec{\mu}\in\Omega_{S}$, the following theorem provides a necessary and sufficient condition for nonlocality in terms of ME with PI.

\begin{theorem}\label{thm:iepi}
For a bipartite quantum state ensemble $\mathcal{E}=\{\eta_{i},\rho_{i}\}_{i\in\Lambda}$, a two-element subset $S$ of $\Lambda$ and $\vec{\mu}\in\Omega_{S}$ satisfying Condition~\eqref{eq:scmpc},
\begin{equation}\label{eq:iepi}
p_{\rm SEP}^{\PI}(\mathcal{E},S)<p_{\rm G}^{\PI}(\mathcal{E},S)
\end{equation}
if and only if there exists an EW in $\{\tilde{\eta}_{\vec{\mu}}\tilde{\rho}_{\vec{\mu}}-\tilde{\eta}_{\vec{\omega}}\tilde{\rho}_{\vec{\omega}}\}_{\vec{\omega}\in\Omega_{S}}$.
\end{theorem}

\begin{proof}
Let us suppose Eq.~\eqref{eq:iepi}.
To show the existence of EW in $\{\tilde{\eta}_{\vec{\mu}}\tilde{\rho}_{\vec{\mu}}-\tilde{\eta}_{\vec{\omega}}\tilde{\rho}_{\vec{\omega}}\}_{\vec{\omega}\in\Omega_{S}}\subseteq\mathbb{SEP}^{*}$, we first assume
\begin{equation}\label{eq:asps}
\tilde\eta_{\vec\mu}\tilde\rho_{\vec\mu}-\tilde\eta_{\vec{\omega}}\tilde\rho_{\vec{\omega}}\in\mathbb{H}_{+}
\end{equation}
for all $\vec{\omega}\in\Omega_{S}$ and lead to a contradiction.
For a measurement $\tilde{\mathcal{M}}=\{\tilde{M}_{\vec\omega}\}_{\vec\omega\in\Omega_{S}}$ providing $p_{\rm G}(\tilde{\mathcal{E}})$, we have 
\begin{eqnarray}
p_{\rm G}^{\rm PI}(\mathcal{E},S)&=&2p_{\rm G}(\tilde{\mathcal{E}})\nonumber\\
&=&2\sum_{\vec\omega\in\Omega_{S}}\tilde\eta_{\vec\omega}\Tr(\tilde\rho_{\vec\omega}\tilde{M}_{\vec\omega})\nonumber\\
&=&2\sum_{\vec\omega\in\Omega_{S}}\Tr\Big[(\tilde\eta_{\vec\mu}\tilde\rho_{\vec\mu}-\tilde\eta_{\vec\mu}\tilde\rho_{\vec\mu}+\tilde\eta_{\vec\omega}\tilde\rho_{\vec\omega})\tilde{M}_{\vec\omega}\Big]\nonumber\\
&=&2\tilde\eta_{\vec\mu}\Tr(\tilde\rho_{\vec\mu}\sum_{\vec\omega\in\Omega_{S}}\tilde{M}_{\vec\omega})
-2\sum_{\vec\omega\in\Omega_{S}}\Tr[(\tilde\eta_{\vec\mu}\tilde\rho_{\vec\mu}-\tilde\eta_{\vec\omega}\tilde\rho_{\vec\omega})\tilde{M}_{\vec\omega}]\nonumber\\
&\leqslant&2\tilde\eta_{\vec\mu}\Tr(\tilde\rho_{\vec\mu}\sum_{\vec\omega\in\Omega_{S}}\tilde{M}_{\vec\omega})\nonumber\\
&=&2\tilde\eta_{\vec\mu}=p_{\rm SEP}^{\rm PI}(\mathcal{E},S),\label{eq:pgue}
\end{eqnarray}
where the first equality is due to Eq.~\eqref{eq:plpie2}, the second equality is from the assumption of $\tilde{\mathcal{M}}$, the inequality follows from Inclusion~\eqref{eq:asps} together with $\tilde{M}_{\vec\omega}\in\mathbb{H}_{+}$ for all $\vec\omega\in\Omega_{S}$, the fifth equality is by $\sum_{\vec\omega\in\Omega_{S}}\tilde{M}_{\vec\omega}=\mathbbm{1}$, and the last equality is from Theorem~\ref{thm:scmp}.
Since Inequality \eqref{eq:pgue} contradicts Inequality~\eqref{eq:iepi}, $\{\tilde{\eta}_{\vec{\mu}}\tilde{\rho}_{\vec{\mu}}-\tilde{\eta}_{\vec{\omega}}\tilde{\rho}_{\vec{\omega}}\}_{\vec{\omega}\in\Omega_{S}}\subseteq\mathbb{SEP}^{*}$ is not a subset of $\mathbb{H}_{+}$. 
Thus, there exists an EW in $\{\tilde{\eta}_{\vec{\mu}}\tilde{\rho}_{\vec{\mu}}-\tilde{\eta}_{\vec{\omega}}\tilde{\rho}_{\vec{\omega}}\}_{\vec{\omega}\in\Omega_{S}}$.

Conversely, let us suppose that $\tilde{\eta}_{\vec{\mu}}\tilde{\rho}_{\vec{\mu}}-\tilde{\eta}_{\vec{\omega}^{\bullet}}\tilde{\rho}_{\vec{\omega}^{\bullet}}$ is an EW for some $\vec\omega^{\bullet}\in\Omega_{S}$. 
Due to this assumption, there exists a pure state $\ket{u}\!\bra{u}\in\mathbb{H}_{+}\backslash\mathbb{SEP}$ satisfying 	
\begin{equation}\label{eq:ewvpz}
\bra{u}(\tilde{\eta}_{\vec{\mu}}\tilde{\rho}_{\vec{\mu}}-\tilde{\eta}_{\vec\omega^{\bullet}}\tilde{\rho}_{\vec\omega^{\bullet}})\ket{u}<0.
\end{equation}
For a measurement $\tilde{\mathcal{M}}=\{\tilde{M}_{\vec\omega}\}_{\vec\omega\in\Omega_{S}}$ with
\begin{equation}\label{eq:amew}
\tilde M_{\vec\mu}=\mathbbm{1}-\ket{u}\!\bra{u},~
\tilde M_{\vec\omega^{\bullet}}=\ket{u}\!\bra{u},
\end{equation} 
we have
\begin{eqnarray}
p_{\rm G}^{\rm PI}(\mathcal{E},S)&=&2p_{\rm G}(\tilde{\mathcal{E}})\nonumber\\
&\geqslant&2\sum_{\vec\omega\in\Omega_{S}}\Tr(\tilde\eta_{\vec\omega}\tilde\rho_{\vec\omega}\tilde{M}_{\vec\omega})\nonumber\\
&=&2\sum_{\vec\omega\in\Omega_{S}}\Tr\Big[(\tilde\eta_{\vec\mu}\tilde\rho_{\vec\mu}-\tilde\eta_{\vec\mu}\tilde\rho_{\vec\mu}+\tilde\eta_{\vec\omega}\tilde\rho_{\vec\omega})\tilde{M}_{\vec\omega}\Big]\nonumber\\
&=&2\tilde\eta_{\vec\mu}\Tr(\tilde\rho_{\vec\mu}\sum_{\vec\omega\in\Omega_{S}}\tilde{M}_{\vec\omega})
-2\sum_{\vec\omega\in\Omega_{S}}\Tr[(\tilde\eta_{\vec\mu}\tilde\rho_{\vec\mu}-\tilde\eta_{\vec\omega}\tilde\rho_{\vec\omega})\tilde{M}_{\vec\omega}]\nonumber\\&=&2\tilde\eta_{\vec\mu}-2\bra{u}(\tilde\eta_{\vec\mu}\tilde\rho_{\vec\mu}-\tilde\eta_{\vec\omega^{\bullet}}\tilde\rho_{\vec\omega^{\bullet}})\ket{u}\nonumber\\[2mm]
&>&2\tilde\eta_{\vec\mu}=p_{\rm SEP}^{\rm PI}(\mathcal{E},S),\label{eq:pgobe}
\end{eqnarray}
where the first equality is by Eq.~\eqref{eq:plpie2}, the first inequality is from the definition of $p_{\rm G}(\tilde{\mathcal{E}})$ in Eq.~\eqref{eq:pgdef}, the fourth equality is due to $\sum_{\vec\omega\in\Omega_{S}}\tilde{M}_{\vec\omega}=\mathbbm{1}$ together with Eq.~\eqref{eq:amew}, and the last inequality is from Inequality~\eqref{eq:ewvpz}, and the last equality is from Theorem~\ref{thm:scmp}.
Thus, Inequality~\eqref{eq:pgobe} leads us to Inequality~\eqref{eq:iepi}.
\end{proof}

\section{Annihilating and creating nonlocality in quantum state discrimination from the choice of postmeasurement information}\label{sec:lun}
In this section, we first recall the definitions of annihilating and creating nonlocality by PI.
By using properties of EW related to nonlocality in discriminating quantum states, we provide four-state ensembles to show PI dependency of annihilating or creating nonlocality in quantum state discrimination.
\begin{definition}
For an ensemble $\mathcal{E}$ in Eq.~\eqref{eq:ens} and a two-element subset $S$ of $\Lambda$, we say that $\PI_{S}$ \emph{annihilates} nonlocality if nonlocality occurs in discriminating the states of $\mathcal{E}$ and
the availability of $\PI_{S}$ vanishes the occurrence of nonlocality, that is,
\begin{equation}\label{eq:annl}
p_{\rm L}(\mathcal{E})<p_{\rm G}(\mathcal{E}),~
p_{\rm L}^{\PI}(\mathcal{E},S)=p_{\rm G}^{\PI}(\mathcal{E},S).
\end{equation}
Also, we say that $\PI_{S}$ \emph{creates} nonlocality if nonlocality does not occur in discriminating the states of $\mathcal{E}$ and
the availability of $\PI_{S}$ releases the occurrence of nonlocality, that is,
\begin{equation}\label{eq:ctnl}
p_{\rm L}(\mathcal{E})=p_{\rm G}(\mathcal{E}),~
p_{\rm L}^{\PI}(\mathcal{E},S)<p_{\rm G}^{\PI}(\mathcal{E},S).
\end{equation}
\end{definition}

The following theorem establishes a sufficient condition for nonlocality arising in discriminating quantum states to be annihilated depending on the choice of subensembles provided by PI.\\
\begin{theorem}\label{thm:alc}
For a bipartite quantum state ensemble $\mathcal{E}=\{\eta_{i},\rho_{i}\}_{i\in\Lambda}$ with $\Lambda=\{1,2,3,4\}$ satisfying
\begin{flalign}
&\eta_{1}\rho_{1}-\eta_{2}\rho_{2}\in\mathbb{H}_{+},\nonumber\\
&\eta_{1}\rho_{1}-\eta_{3}\rho_{3}\in\mathbb{SEP}^{*}\backslash\mathbb{H}_{+},\nonumber\\
&\eta_{2}\rho_{2}-\eta_{4}\rho_{4}\in\mathbb{SEP}^{*},\nonumber\\
&\eta_{3}\rho_{3}-\eta_{4}\rho_{4}\in\mathbb{H}_{+},
\label{eq:alc}
\end{flalign}
$\PI_{\{1,2\}}$ annihilates nonlocality but $\PI_{\{1,3\}}$ does not annihilate nonlocality.
\end{theorem}

\begin{proof}
We first show that quantum nonlocality occurs in discriminating the states of $\mathcal{E}$. 
Due to $\mathbb{H}_{+}\subseteq\mathbb{SEP}^{*}$ and the definition of EW in Eq.~\eqref{eq:wsshp}, Condition~\eqref{eq:alc} implies that $\eta_{1}\rho_{1}-\eta_{2}\rho_{2}$ and $\eta_{1}\rho_{1}-\eta_{3}\rho_{3}$ are in $\mathbb{SEP}^{*}$ and $\eta_{1}\rho_{1}-\eta_{3}\rho_{3}$ is an EW.
Since $\mathbb{SEP}^{*}$ is a convex cone, the sum of $\eta_{1}\rho_{1}-\eta_{2}\rho_{2}$ and $\eta_{2}\rho_{2}-\eta_{4}\rho_{4}$ is in $\mathbb{SEP}^{*}$, that is,
\begin{equation}\label{eq:er14}
\eta_{1}\rho_{1}-\eta_{4}\rho_{4}=(\eta_{1}\rho_{1}-\eta_{2}\rho_{2})+(\eta_{2}\rho_{2}-\eta_{4}\rho_{4})\in\mathbb{SEP}^{*}.
\end{equation}
Thus, Proposition~\ref{prop:sgi} leads us to $p_{\rm SEP}(\mathcal{E})<p_{\rm G}(\mathcal{E})$.

Now, we show that $\PI_{\{1,2\}}$ annihilates nonlocality but $\PI_{\{1,3\}}$ does not annihilate nonlocality.
For $S=\{1,2\}$ and $\vec\mu=(1,3)$, we have
\begin{eqnarray}\label{eq:lsa}
\tilde\eta_{\vec\mu}\tilde\rho_{\vec\mu}-\tilde\eta_{(1,4)}\tilde\rho_{(1,4)}&=&\tfrac{1}{2}(\eta_{3}\rho_{3}-\eta_{4}\rho_{4})\in\mathbb{H}_{+},\nonumber\\
\tilde\eta_{\vec\mu}\tilde\rho_{\vec\mu}-\tilde\eta_{(2,3)}\tilde\rho_{(2,3)}&=&\tfrac{1}{2}(\eta_{1}\rho_{1}-\eta_{2}\rho_{2})\in\mathbb{H}_{+},\nonumber\\
\tilde\eta_{\vec\mu}\tilde\rho_{\vec\mu}-\tilde\eta_{(2,4)}\tilde\rho_{(2,4)}&=&\tfrac{1}{2}(\eta_{1}\rho_{1}-\eta_{2}\rho_{2})+\tfrac{1}{2}(\eta_{3}\rho_{3}-\eta_{4}\rho_{4})\in\mathbb{H}_{+},
\end{eqnarray}
where the equalities are by the definition of $\{\tilde\eta_{\vec\omega},\tilde\rho_{\vec\omega}\}_{\vec\omega\in\Omega_{S}}$ in Eqs.~\eqref{eq:avpr} and \eqref{eq:avrh} and the inclusions are from  Condition~\eqref{eq:alc} together with the fact that $\mathbb{H}_{+}$ is a convex cone. 
Since Inclusion~\eqref{eq:lsa} implies Condition~\eqref{eq:scmpc}, 
it follows from Theorem~\ref{thm:iepi} and the definition of EW in Eq.~\eqref{eq:wsshp} that $p_{\rm SEP}^{\rm PI}(\mathcal{E},S)=p_{\rm G}^{\rm PI}(\mathcal{E},S)$.
Thus, $\PI_{\{1,2\}}$ annihilates nonlocality.

For $S=\{1,3\}$ and $\vec\mu=(1,2)$, we have
\begin{eqnarray}\label{eq:lsb}
\tilde\eta_{\vec\mu}\tilde\rho_{\vec\mu}-\tilde\eta_{(1,4)}\tilde\rho_{(1,4)}&=&\tfrac{1}{2}(\eta_{2}\rho_{2}-\eta_{4}\rho_{4})\in\mathbb{SEP}^{*},\nonumber\\
\tilde\eta_{\vec\mu}\tilde\rho_{\vec\mu}-\tilde\eta_{(3,2)}\tilde\rho_{(3,2)}&=&\tfrac{1}{2}(\eta_{1}\rho_{1}-\eta_{3}\rho_{3})\in\mathbb{SEP}^{*}\backslash\mathbb{H}_{+},\nonumber\\
\tilde\eta_{\vec\mu}\tilde\rho_{\vec\mu}-\tilde\eta_{(3,4)}\tilde\rho_{(3,4)}&=&\tfrac{1}{2}(\eta_{1}\rho_{1}-\eta_{3}\rho_{3})+\tfrac{1}{2}(\eta_{2}\rho_{2}-\eta_{4}\rho_{4})\in\mathbb{SEP}^{*},
\end{eqnarray}
where the equalities are from the definition of $\{\tilde\eta_{\vec\omega},\tilde\rho_{\vec\omega}\}_{\vec\omega\in\Omega_{S}}$ in Eqs.~\eqref{eq:avpr} and \eqref{eq:avrh} and the inclusions follow from  Condition~\eqref{eq:alc} together with the fact that $\mathbb{SEP}^{*}$ is a convex cone. 
Due to Inclusion~\eqref{eq:lsb} and the definition of EW in Eq.~\eqref{eq:wsshp}, Condition~\eqref{eq:scmpc} is satisfied and $\tilde\eta_{\vec\mu}\tilde\rho_{\vec\mu}-\tilde\eta_{\vec\omega}\tilde\rho_{\vec\omega}$ is an EW for $\vec\omega=(3,2)$.
Theorem~\ref{thm:iepi} leads us to $p_{\rm SEP}^{\rm PI}(\mathcal{E},S)<p_{\rm G}^{\rm PI}(\mathcal{E},S)$.
Thus, $\PI_{\{1,3\}}$ does not annihilate nonlocality.
\end{proof}

Now, we provide bipartite quantum state ensembles $\mathcal{E}=\{\eta_{i},\rho_{i}\}_{i\in\Lambda}$ with $\Lambda=\{1,2,3,4\}$ where $\PI_{\{1,2\}}$ annihilates nonlocality but $\PI_{\{1,3\}}$ does not annihilate nonlocality. 

\begin{example}\label{ex:anpi}
For an EW $W$, let us consider the ensemble $\mathcal{E}=\{\eta_{i},\rho_{i}\}_{i\in\Lambda}$ consisting of
\begin{flalign}
\eta_{1}=\tfrac{\Tr(2W_{+}+W_{-})}{4\Tr(W_{+}+W_{-})},~&\rho_{1}=\tfrac{2W_{+}+W_{-}}{\Tr(2W_{+}+W_{-})},\nonumber\\
\eta_{2}=\tfrac{\Tr W_{+}}{4\Tr(W_{+}+W_{-})},~&\rho_{2}=\tfrac{W_{+}}{\Tr W_{+}},\nonumber\\
\eta_{3}=\tfrac{\Tr(W_{+}+2W_{-})}{4\Tr(W_{+}+W_{-})},~&\rho_{3}=\tfrac{W_{+}+2W_{-}}{\Tr(W_{+}+2W_{-})},\nonumber\\
\eta_{4}=\tfrac{\Tr W_{-}}{4\Tr(W_{+}+W_{-})},~&\rho_{4}=\tfrac{W_{-}}{\Tr W_{-}}\label{eq:alex}
\end{flalign}
where $W_{\pm}$ is the positive-semidefinite operator satisfying
\begin{equation}\label{eq:wpmo}
\Tr(W_{+}W_{-})=0,~W=W_{+}-W_{-}.
\end{equation}
\end{example}

A straightforward calculation leads us to
\begin{eqnarray}
\eta_{1}\rho_{1}-\eta_{2}\rho_{2}=\eta_{3}\rho_{3}-\eta_{4}\rho_{4}=\tfrac{W_{+}+W_{-}}{4\Tr(W_{+}+W_{-})},\nonumber\\
\eta_{2}\rho_{2}-\eta_{4}\rho_{4}=\eta_{1}\rho_{1}-\eta_{3}\rho_{3}=\tfrac{W}{4\Tr(W_{+}+W_{-})},\label{eq:wpmm}
\end{eqnarray}
which imply Condition~\eqref{eq:alc} in Theorem~\ref{thm:alc}.
Thus, $\PI_{\{1,2\}}$ annihilates nonlocality but $\PI_{\{1,3\}}$ does not annihilate nonlocality.

The following theorem establishes a sufficient condition for nonlocality arising in discriminating quantum states to be created depending on the choice of subensembles provided by PI.
\begin{theorem}\label{thm:aulc}
For a bipartite quantum state ensemble $\mathcal{E}=\{\eta_{i},\rho_{i}\}_{i\in\Lambda}$ with $\Lambda=\{1,2,3,4,\}$ satisfying
\begin{flalign}
&\eta_{1}\rho_{1}-\eta_{2}\rho_{2}\in\mathbb{H}_{+},\nonumber\\
&\eta_{1}\rho_{1}-\eta_{3}\rho_{3}\in\mathbb{H}_{+},\nonumber\\
&\eta_{2}\rho_{2}-\eta_{4}\rho_{4}\in\mathbb{H}_{+},\nonumber\\
&\eta_{3}\rho_{3}-\eta_{4}\rho_{4}\in\mathbb{SEP}^{*}\backslash\mathbb{H}_{+},
\label{eq:aulc}
\end{flalign}
$\PI_{\{1,2\}}$ creates nonlocality but $\PI_{\{1,3\}}$ does not create nonlocality.
\end{theorem}

\begin{proof}
We first show that quantum nonlocality does not occur in discriminating the states of $\mathcal{E}$. 
Since $\mathbb{H}_{+}$ is a convex cone, the sum of $\eta_{1}\rho_{1}-\eta_{2}\rho_{2}$ and $\eta_{2}\rho_{2}-\eta_{4}\rho_{4}$ is in $\mathbb{H}_{+}$, that is,
\begin{equation}\label{eq:er14h}
\eta_{1}\rho_{1}-\eta_{4}\rho_{4}=(\eta_{1}\rho_{1}-\eta_{2}\rho_{2})+(\eta_{2}\rho_{2}-\eta_{4}\rho_{4})\in\mathbb{H}_{+}.
\end{equation}
Thus, Proposition~\ref{prop:sgi} and the definition of EW in Eq.~\eqref{eq:wsshp} together with Eqs.~\eqref{eq:aulc} and \eqref{eq:er14h} lead us to $p_{\rm SEP}(\mathcal{E})=p_{\rm G}(\mathcal{E})$.

Now, we show that $\PI_{\{1,2\}}$ creates nonlocality but $\PI_{\{1,3\}}$ does not create nonlocality.
For $S=\{1,2\}$ and $\vec\mu=(1,3)$, we have
\begin{eqnarray}\label{eq:ul12}
\tilde\eta_{\vec\mu}\tilde\rho_{\vec\mu}-\tilde\eta_{(1,4)}\tilde\rho_{(1,4)}&=&\tfrac{1}{2}(\eta_{3}\rho_{3}-\eta_{4}\rho_{4})\in\mathbb{SEP}^{*}\backslash\mathbb{H}_{+},\nonumber\\
\tilde\eta_{\vec\mu}\tilde\rho_{\vec\mu}-\tilde\eta_{(2,3)}\tilde\rho_{(2,3)}&=&\tfrac{1}{2}(\eta_{1}\rho_{1}-\eta_{2}\rho_{2})\in\mathbb{SEP}^{*},\nonumber\\
\tilde\eta_{\vec\mu}\tilde\rho_{\vec\mu}-\tilde\eta_{(2,4)}\tilde\rho_{(2,4)}&=&\tfrac{1}{2}(\eta_{1}\rho_{1}-\eta_{2}\rho_{2})+\tfrac{1}{2}(\eta_{3}\rho_{3}-\eta_{4}\rho_{4})\in\mathbb{SEP}^{*},
\end{eqnarray}
where the equalities are from the definition of $\{\tilde\eta_{\vec\omega},\tilde\rho_{\vec\omega}\}_{\vec\omega\in\Omega_{S}}$ in Eqs.~\eqref{eq:avpr} and \eqref{eq:avrh} and the inclusions follow from Condition~\eqref{eq:aulc} and $\mathbb{H}_{+}\subseteq\mathbb{SEP}^{*}$ together with the fact that $\mathbb{SEP}^{*}$ is a convex cone. 
Due to Inclusion~\eqref{eq:ul12} and the definition of EW in Eq.~\eqref{eq:wsshp}, Condition~\eqref{eq:scmpc} is satisfied and $\tilde\eta_{\vec\mu}\tilde\rho_{\vec\mu}-\tilde\eta_{\vec\omega}\tilde\rho_{\vec\omega}$ is an EW for $\vec\omega=(1,4)$.
Theorem~\ref{thm:iepi} leads us to $p_{\rm SEP}^{\rm PI}(\mathcal{E},S)<p_{\rm G}^{\rm PI}(\mathcal{E},S)$.
Thus, $\PI_{\{1,2\}}$ creates nonlocality.

For $S=\{1,3\}$ and $\vec\mu=(1,2)$, we have
\begin{eqnarray}\label{eq:ul13}
\tilde\eta_{\vec\mu}\tilde\rho_{\vec\mu}-\tilde\eta_{(1,4)}\tilde\rho_{(1,4)}&=&\tfrac{1}{2}(\eta_{2}\rho_{2}-\eta_{4}\rho_{4})\in\mathbb{H}_{+},\nonumber\\
\tilde\eta_{\vec\mu}\tilde\rho_{\vec\mu}-\tilde\eta_{(3,2)}\tilde\rho_{(3,2)}&=&\tfrac{1}{2}(\eta_{1}\rho_{1}-\eta_{3}\rho_{3})\in\mathbb{H}_{+},\nonumber\\
\tilde\eta_{\vec\mu}\tilde\rho_{\vec\mu}-\tilde\eta_{(3,4)}\tilde\rho_{(3,4)}&=&\tfrac{1}{2}(\eta_{1}\rho_{1}-\eta_{3}\rho_{3})+\tfrac{1}{2}(\eta_{2}\rho_{2}-\eta_{4}\rho_{4})\in\mathbb{H}_{+},
\end{eqnarray}
where the equalities are by the definition of $\{\tilde\eta_{\vec\omega},\tilde\rho_{\vec\omega}\}_{\vec\omega\in\Omega_{S}}$ in Eqs.~\eqref{eq:avpr} and \eqref{eq:avrh} and the inclusions are from Condition~\eqref{eq:aulc} together with the fact that $\mathbb{H}_{+}$ is a convex cone. 
Since Inclusion~\eqref{eq:ul13} implies Condition~\eqref{eq:scmpc}, it follows from Theorem~\ref{thm:iepi} and the definition of EW in Eq.~\eqref{eq:wsshp} that $p_{\rm SEP}^{\rm PI}(\mathcal{E},S)=p_{\rm G}^{\rm PI}(\mathcal{E},S)$.
Thus, $\PI_{\{1,3\}}$ does not create nonlocality.
\end{proof}

Now, we provide bipartite quantum state ensembles $\mathcal{E}=\{\eta_{i},\rho_{i}\}_{i\in\Lambda}$ with $\Lambda=\{1,2,3,4\}$ where $\PI_{\{1,2\}}$ creates nonlocality but $\PI_{\{1,3\}}$ does not create nonlocality. 

\begin{example}\label{ex:crpi}
For an EW $W$, let us consider the ensemble $\mathcal{E}=\{\eta_{i},\rho_{i}\}_{i\in\Lambda}$ consisting of
\begin{flalign}\label{eq:aulex}
\eta_{1}=\tfrac{\Tr(2W_{+}+W_{-})}{\Tr(4W_{+}+3W_{-})},~&\rho_{1}=\tfrac{2W_{+}+W_{-}}{\Tr(2W_{+}+W_{-})},\nonumber\\
\eta_{2}=\tfrac{\Tr (W_{+}+W_{-})}{\Tr(4W_{+}+3W_{-})},~&\rho_{2}=\tfrac{W_{+}+W_{-}}{\Tr (W_{+}+W_{-})},\nonumber\\
\eta_{3}=\tfrac{\Tr W_{+}}{\Tr(4W_{+}+3W_{-})},~&\rho_{3}=\tfrac{W_{+}}{\Tr W_{+}},\nonumber\\
\eta_{4}=\tfrac{\Tr W_{-}}{\Tr(4W_{+}+3W_{-})},~&\rho_{4}=\tfrac{W_{-}}{\Tr W_{-}},
\end{flalign}
where $W_{\pm}$ is the positive-semidefinite operator satisfying
Eq.~\eqref{eq:wpmo}.
\end{example}

From a straightforward calculation, we can verify that
\begin{flalign}\label{eq:ewpmm}
\eta_{1}\rho_{1}-\eta_{2}\rho_{2}&=\tfrac{W_{+}}{\Tr(4W_{+}+3W_{-})},\nonumber\\
\eta_{1}\rho_{1}-\eta_{3}\rho_{3}&=\tfrac{W_{+}+W_{-}}{\Tr(4W_{+}+3W_{-})},\nonumber\\
\eta_{2}\rho_{2}-\eta_{4}\rho_{4}&=\tfrac{W_{-}}{\Tr(4W_{+}+3W_{-})},\nonumber\\
\eta_{3}\rho_{3}-\eta_{4}\rho_{4}&=\tfrac{W}{\Tr(4W_{+}+3W_{-})},
\end{flalign}
which imply Condition~\eqref{eq:aulc} in Theorem~\ref{thm:aulc}.
Thus, $\PI_{\{1,2\}}$ creates nonlocality but $\PI_{\{1,3\}}$ does not create nonlocality.

\section{Discussion}\label{sec:dis}
We have investigated bipartite quantum state discrimination and shown that the annihilation or creation of nonlocality by PI can depend on the choice of subensembles. We first established sufficient conditions for nonlocality in ME with PI (Theorems~\ref{thm:scmp} and \ref{thm:iepi}). We then derived sufficient conditions under which the annihilation or creation of nonlocality depends on how the original ensemble is partitioned into subensembles (Theorems~\ref{thm:alc} and \ref{thm:aulc}). Finally, we constructed explicit bipartite quantum state ensembles satisfying these conditions, thereby illustrating the partition-dependent nature of nonlocality in state discrimination (Examples~\ref{ex:anpi} and \ref{ex:crpi}).

The ensembles $\mathcal{E}=\{\eta_{i},\rho_{i}\}_{i\in\Lambda}$ presented in Examples~\ref{ex:anpi} and \ref{ex:crpi} necessarily contain an entangled state. Indeed, the expectation value of the entanglement witness $W$ in Eq.~\eqref{eq:wpmo} with respect to $\rho_4$ is always negative. By the definition of an entanglement witness, this implies that $\rho_4$ is entangled. Nevertheless, there exist ensembles consisting entirely of separable states that satisfy the conditions of Theorem~\ref{thm:alc}.
\begin{example}\label{ex:nda}
Let us consider the two-qubit state ensemble $\mathcal{E}=\{\eta_{i},\rho_{i}\}_{i\in\Lambda}$ consisting of separable states
\begin{flalign}\label{eq:sec}
\eta_{1}=\tfrac{3}{7},~\rho_{1}=&\tfrac{1}{3}\ket{+}\!\bra{+}\otimes\ket{+}\!\bra{+}+\tfrac{1}{3}\ket{-}\!\bra{-}\otimes\ket{-}\!\bra{-}
+\tfrac{1}{12}\mathbbm{1},\nonumber\\[1mm]
\eta_{2}=\tfrac{3}{14},~\rho_{2}=&\tfrac{1}{6}\ket{0}\!\bra{0}\otimes\ket{0}\!\bra{0}
+\tfrac{1}{6}\ket{1}\!\bra{1}\otimes\ket{1}\!\bra{1}\nonumber\\
&+\tfrac{1}{3}\ket{+}\!\bra{+}\otimes\ket{+}\!\bra{+}
+\tfrac{1}{3}\ket{-}\!\bra{-}\otimes\ket{-}\!\bra{-},
\nonumber\\[1mm]
\eta_{3}=\tfrac{3}{14},~\rho_{3}=&\tfrac{1}{6}\ket{0}\!\bra{0}\otimes\ket{0}\!\bra{0}
+\tfrac{1}{6}\ket{1}\!\bra{1}\otimes\ket{1}\!\bra{1}\nonumber\\
&+\tfrac{1}{6}\ket{\xi_{+}}\!\bra{\xi_{+}}\otimes\ket{\xi_{+}}\!\bra{\xi_{+}}
+\tfrac{1}{6}\ket{\xi_{-}}\!\bra{\xi_{-}}\otimes\ket{\xi_{-}}\!\bra{\xi_{-}}\nonumber\\
&+\tfrac{1}{6}\ket{\zeta_{+}}\!\bra{\zeta_{+}}\otimes\ket{\zeta_{+}}\!\bra{\zeta_{+}}
+\tfrac{1}{6}\ket{\zeta_{-}}\!\bra{\zeta_{-}}\otimes\ket{\zeta_{-}}\!\bra{\zeta_{-}},\nonumber\\[1mm]
\eta_{4}=\tfrac{1}{7},~\rho_{4}=&\tfrac{1}{2}\ket{+}\!\bra{+}\otimes\ket{+}\!\bra{+}
+\tfrac{1}{2}\ket{-}\!\bra{-}\otimes\ket{-}\!\bra{-}
\end{flalign}
where
\begin{flalign}\label{eq:dss}
\ket{\pm}=&\tfrac{1}{\sqrt{2}}(\ket{0}\pm\ket{1}),\nonumber\\
\ket{\xi_{\pm}}=&\tfrac{1}{\sqrt{2}}(\ket{0}\pm e^{\mathrm{i}\pi/4}\ket{1}),\nonumber\\
\ket{\zeta_{\pm}}=&\tfrac{1}{\sqrt{2}}(\ket{0}\pm e^{-\mathrm{i}\pi/4}\ket{1}).
\end{flalign}
\end{example}

From a straightforward calculation, we can verify that
\begin{flalign}\label{eq:sfc}
\eta_{1}\rho_{1}-\eta_{2}\rho_{2}=&\tfrac{1}{14}\Phi_{+}
+\tfrac{1}{14}\Psi_{+}
+\tfrac{1}{28}\ket{01}\!\bra{01}
+\tfrac{1}{28}\ket{10}\!\bra{10}\in\mathbb{H}_{+},
\nonumber\\[1mm]
\eta_{1}\rho_{1}-\eta_{3}\rho_{3}
=&\tfrac{1}{14}\Phi_{+}^{\PT}
+\tfrac{1}{7}\Psi_{+}^{\PT}\in\mathbb{SEP}^{*}\backslash\mathbb{H}_{+},
\nonumber\\[1mm]
\eta_{2}\rho_{2}-\eta_{4}\rho_{4}=&\tfrac{1}{28}\ket{00}\!\bra{00}+\tfrac{1}{28}\ket{11}\!\bra{11}\in\mathbb{SEP}^{*},
\nonumber\\[1mm]
\eta_{3}\rho_{3}-\eta_{4}\rho_{4}=&\tfrac{1}{14}\Phi_{-}\in\mathbb{H}_{+}
\end{flalign}
where $\PT$ is the partial transposition taken in the standard basis $\{\ket{0},\ket{1}\}$ on the second subsystem and
\begin{flalign}\label{eq:debs}
\Phi_{\pm}=\ket{\Phi_{\pm}}\!\bra{\Phi_{\pm}},&
\ket{\Phi_{\pm}}=\tfrac{1}{\sqrt{2}}\ket{00}\pm\tfrac{1}{\sqrt{2}}\ket{11},\nonumber\\
\Psi_{\pm}=\ket{\Psi_{\pm}}\!\bra{\Psi_{\pm}},&
\ket{\Psi_{\pm}}=\tfrac{1}{\sqrt{2}}\ket{01}\pm\tfrac{1}{\sqrt{2}}\ket{10}.
\end{flalign}
The inclusion $\eta_{1}\rho_{1}-\eta_{3}\rho_{3}\in\mathbb{SEP}^{*}\backslash\mathbb{H}_{+}$ in Eq.~\eqref{eq:sfc} follows from $\bra{\Phi_{-}}(\eta_{1}\rho_{1}-\eta_{3}\rho_{3})\ket{\Phi_{-}}=-1/28$ and $E^{\PT}\in\mathbb{SEP}^{*}$ for all $E\in\mathbb{H}_{+}$.
The inclusion $\eta_{2}\rho_{2}-\eta_{4}\rho_{4}\in\mathbb{SEP}^{*}$ in Eq.~\eqref{eq:sfc} is by $\mathbb{H}_{+}\subset\mathbb{SEP}^{*}$.
Since Eq.~\eqref{eq:sfc} implies Condition~\eqref{eq:alc} in Theorem~\ref{thm:alc}, $\PI_{\{1,2\}}$ annihilates nonlocality but $\PI_{\{1,3\}}$ does not annihilate nonlocality.

Moreover, there exist ensembles consisting entirely of separable states that satisfy the conditions of Theorem~\ref{thm:aulc}.
\begin{example}\label{ex:ndc}
Let us consider the two-qubit state ensemble $\mathcal{E}=\{\eta_{i},\rho_{i}\}_{i\in\Lambda}$ consisting of separable states
\begin{flalign}\label{eq:sea}
\eta_{1}=\tfrac{1}{3},~\rho_{1}=&\tfrac{1}{4}\ket{+}\!\bra{+}\otimes\ket{+}\!\bra{+}
+\tfrac{1}{4}\ket{-}\!\bra{-}\otimes\ket{-}\!\bra{-}
+\tfrac{1}{8}\mathbbm{1},
\nonumber\\[1mm]
\eta_{2}=\tfrac{1}{4},~\rho_{2}=&\tfrac{1}{6}\ket{0}\!\bra{0}\otimes\ket{1}\!\bra{1}
+\tfrac{1}{6}\ket{1}\!\bra{1}\otimes\ket{0}\!\bra{0}\nonumber\\
&+\tfrac{1}{3}\ket{+}\!\bra{+}\otimes\ket{+}\!\bra{+}
+\tfrac{1}{3}\ket{-}\!\bra{-}\otimes\ket{-}\!\bra{-},
\nonumber\\[1mm]
\eta_{3}=\tfrac{1}{4},~\rho_{3}=&\tfrac{1}{6}\ket{0}\!\bra{0}\otimes\ket{0}\!\bra{0}
+\tfrac{1}{6}\ket{1}\!\bra{1}\otimes\ket{1}\!\bra{1}\nonumber\\
&+\tfrac{1}{3}\ket{+}\!\bra{+}\otimes\ket{+}\!\bra{+}
+\tfrac{1}{3}\ket{-}\!\bra{-}\otimes\ket{-}\!\bra{-},\nonumber\\[1mm]
\eta_{4}=\tfrac{1}{6},~\rho_{4}=&
\tfrac{1}{4}\ket{\xi_{+}}\!\bra{\xi_{+}}\otimes\ket{\zeta_{+}}\!\bra{\zeta_{+}}
+\tfrac{1}{4}\ket{\xi_{-}}\!\bra{\xi_{-}}\otimes\ket{\zeta_{-}}\!\bra{\zeta_{-}}\nonumber\\
&+\tfrac{1}{4}\ket{\zeta_{+}}\!\bra{\zeta_{+}}\otimes\ket{\xi_{+}}\!\bra{\xi_{+}}
+\tfrac{1}{4}\ket{\zeta_{-}}\!\bra{\zeta_{-}}\otimes\ket{\xi_{-}}\!\bra{\xi_{-}}.
\end{flalign}
\end{example}

It is straightforward to verify that
\begin{flalign}\label{eq:sfa}
\eta_{1}\rho_{1}-\eta_{2}\rho_{2}
=&\tfrac{1}{24}\ket{00}\!\bra{00}
+\tfrac{1}{24}\ket{11}\!\bra{11}\in\mathbb{H}_{+},
\nonumber\\
\eta_{1}\rho_{1}-\eta_{3}\rho_{3}
=&\tfrac{1}{24}\ket{01}\!\bra{01}
+\tfrac{1}{24}\ket{10}\!\bra{10}\in\mathbb{H}_{+},
\nonumber\\
\eta_{2}\rho_{2}-\eta_{4}\rho_{4}=&\tfrac{1}{12}\Psi_{+}\in\mathbb{H}_{+},
\nonumber\\
\eta_{3}\rho_{3}-\eta_{4}\rho_{4}=&\tfrac{1}{12}\Phi_{+}^{\PT}\in\mathbb{SEP}^{*}\backslash\mathbb{H}_{+}.
\end{flalign}
The inclusion $\eta_{3}\rho_{3}-\eta_{4}\rho_{4}\in\mathbb{SEP}^{*}\backslash\mathbb{H}_{+}$ in Eq.~\eqref{eq:sfa} is from $\bra{\Psi_{-}}(\eta_{3}\rho_{3}-\eta_{4}\rho_{4})\ket{\Psi_{-}}=-1/24$ and $E^{\PT}\in\mathbb{SEP}^{*}$ for all $E\in\mathbb{H}_{+}$.
Since Eq.~\eqref{eq:sfa} implies Condition~\eqref{eq:aulc} in Theorem~\ref{thm:aulc}, $\PI_{\{1,2\}}$ creates nonlocality but $\PI_{\{1,3\}}$ does not create nonlocality.

All previously known examples of the annihilation or creation of nonlocality by PI involved ensembles consisting entirely of separable states \cite{ha20221,ha20222}. By contrast, Examples~\ref{ex:anpi} and \ref{ex:crpi} demonstrate that these phenomena can also occur for ensembles containing entangled states. Thus, the presence of entanglement in a state ensemble does not preclude the annihilation or creation of nonlocality by PI. Moreover, Examples~\ref{ex:nda} and \ref{ex:ndc}, together with Examples~\ref{ex:anpi} and \ref{ex:crpi}, indicate that the dependence of these phenomena on the choice of subensembles cannot be determined solely by whether the ensemble contains entangled states.

Our results show that, in ME, the annihilation or creation of quantum nonlocality can depend on the choice of PI. It is therefore natural to investigate whether analogous phenomena arise under other state-discrimination strategies.

\section*{Acknowledgments}
This work was supported by Korea Research Institute for defense Technology planning and advancement (KRIT) grant funded by Defense Acquisition Program Administration(DAPA)(KRIT-CT-23–031) and the Institute for Information \& Communications Technology Planning \& Evaluation(IITP) grant funded by the Korean government(MSIP)(Grant No. RS-2025-02304540). JH and JSK were supported by Creation of the Quantum Information Science R\&D Ecosystem(Grant No. 2022M3H3A106307411) through the National Research Foundation of Korea(NRF) funded by the Korean government(Ministry of Science and ICT).


\end{document}